\documentclass[12pt]{article}
\usepackage{amsgen, amsmath, amsfonts, amsthm, 
amssymb, enumerate,amscd,color, multirow}
\usepackage[all]{xy}

\def\s#1#2{\langle \,#1 , #2 \,\rangle}

\definecolor{blue}{rgb}{0,0,1}
\definecolor{red}{rgb}{1,0,0}
\definecolor{green}{rgb}{0,1,0}

\newtheorem{defn}{Definition}[section]
\newtheorem{prop}[defn]{Proposition}
\newtheorem{lem}[defn]{Lemma}
\newtheorem{thm}[defn]{Theorem}

\newtheorem{rem}[defn]{Remark}

\newcommand {\C}{{\mathbb C}}

\newcommand {\G}{{\Gamma}}

\newcommand {\g}{{\gamma}}

\newcommand {\HH}{{\mathfrak  H}}

\newcommand {\ba}{{\mathfrak a}}
\newcommand {\bb}{{\mathfrak b}}

\newcommand {\ca}{{\mathbf a}}

\def\ca{{\mathfrak a}}

\def\c1{{\mathfrak a_1}}

\def\cb{{\mathfrak b}}
\def\cc{{\mathfrak c}}

\def\sa{{\sigma_\mathfrak a}}

\def\s1{{{\sigma_{\mathfrak a_1}}}}

\def\sb{{\sigma_\mathfrak b}}
\def\sc{{\sigma_\mathfrak c}}

\title{\bf New percolation crossing formulas and second-order modular forms}
\author {Nikolaos Diamantis  \\ \it
\it \small School of Mathematical Sciences\\ \small \it University of Nottingham\\
\it  \small Nottingham, UK 
 \\ and  \\ Peter Kleban \\ 
\it \small LASST and Department of Physics \& Astronomy\\ \it  \small University of Maine, Orono, Maine, USA}
\begin{document}

\maketitle
\abstract   
We consider the three new crossing probabilities for percolation recently found via conformal field theory \cite{SKZ}. We prove that all three of them (i) may be simply expressed in terms of Cardy's \cite{C} and Watts' \cite{W} crossing probabilities, (ii) are (weakly holomorphic) second-order modular forms of weight 0 (and a single particular type) on the congruence group $\Gamma(2)$, and (iii) under some technical assumptions (similar to those used in \cite{KZ}), are completely determined by their transformation laws. 

The only physical input in (iii) is Cardy's crossing formula, which suggests an unknown connection between all crossing-type formulas.
 \endabstract 
 
\section{Introduction}

Modular behavior appears in percolation theory on a rectangle, as is 
demonstrated in \cite{KZ}.  There, formulas for the crossing probabilities at the percolation point obtained via conformal field theory (or more exactly,  their derivatives 
with respect to the aspect ratio $r$) are shown to have interesting 
modular properties.  Further, a modular characterization of these 
probabilities is given via several theorems.

In that work, there are two main quantities of interest.  The 
``horizontal" crossing probability $\Pi_h(r)$ is the probability of 
finding at least one cluster that connects the left and right (vertical) edges of the 
rectangle, with $r$ being the ratio of horizontal to vertical edge lengths.  An 
explicit formula for this quantity was originally 
calculated by Cardy \cite{C} via conformal field theory, and more recently 
proven rigorously, for the triangular lattice, by Smirnov \cite{S}. The 
``horizontal but not vertical" crossing probability $\Pi_{h \bar v}(r)$ is 
the probability of finding at least one cluster that connects the left and right 
edges of the rectangle while the top and bottom 
(horizontal) edges are not connected.  A formula for this was first obtained with conformal 
field theory by Watts \cite{W} and recently proven rigorously by 
Dub{\'e}dat \cite{D} using Schramm-Loewner Evolution methods.

The behavior of the crossing probabilities under $S: z \to -1/z$ (with $z 
= i \, r$) follows from physical symmetries of the problem, while the 
behavior under $T: z \to z +1$ arises from the structure of the formulas 
themselves (explicitly, the appearance of a single ``conformal block", see 
\cite{KZ} for details), but has no obvious physical origin.  

Two surprises arise in \cite{KZ}.  First, there is no reason to expect 
modular behavior  on a rectangle at all  since  it lacks the appropriate symmetry.  (By contrast, on a torus 
modular behavior of crossing probabilities is both expected and 
observed (\cite{P})). Despite this,  
$\Pi'_{h}(r)$ is a modular form (the prime denotes differentiation), and is completely determined by a 
simple modular argument that assumes its physical symmetry.  
Secondly,  $\Pi'_{h \bar v}(r)$  is observed to possess unusual 
modular behavior, which leads to the definition of a new kind of modular 
object, the $n^{th}-order\, modular \, form$.  Such objects were also 
defined independently, in a different context (\cite{CDO}) 
and have since been explored systematically (e.g. \cite{IM}, 
\cite{CD}, \cite{DSr}, \cite{DSi}, \cite{De}).

For percolation, the appearance of a second-order modular form may be 
traced to the difference in sign of the behavior under $S$ of $\Pi'_h$ and 
$\Pi'_{hv}$, where $\Pi_{hv}(r) = \Pi_{h}(r) - \Pi_{h \bar v}(r)$ is the 
probability of finding a cluster that crosses both horizontally and 
vertically.

Recently, three new percolation crossing probabilities have been 
calculated from conformal field theory, $\pi_h^b(\alpha, \beta)$, 
$\pi_h^{\bar b}(\alpha, \beta)$, and $\nu_h(\alpha, \beta)$  (\cite{SKZ}).  
These quantities (see below for the explicit formulas) are more precisely 
probability densities for clusters that connect to two specified points 
 on opposite vertical edges of a 
rectangle, under certain specified conditions (see \cite{SKZ} for 
details, but note that the points on the rectangle are the images of  $\alpha$ and $\beta$, which lie on the real axis).  Their integrals give $\Pi_{h}$, $\Pi_{h \bar v}$, and ${\cal 
N}_h$, where the latter is the expected number of horizontally crossing 
clusters. (This was calculated first via conformal field theory \cite{C2,C3} and later rigorously \cite{HS}.)

In this work we examine the modular properties of $\pi_h^b$, $\pi_h^{\bar 
b}$, and $\nu_h$ under the congruence group $\Gamma(2)$.  Unexpectedly, we 
find that all three of them are second-order modular forms.  This occurs 
even though they again are on a rectangle, which does 
not possess the requisite symmetry under $S$, and they specify crossings 
between points, rather than intervals. These functions are also very 
interesting from a purely number-theoretic viewpoint, as they illustrate 
the occurence of an extension of higher order forms which involves  
\it two \rm group actions and whose study, though formally natural, would 
seem otherwise unmotivated. 

Proposition \ref{PropDZ} shows that all three crossing probabilities may 
be written as expressions linear in the ratio $\phi$, which is 
proportional to $\Pi'_{h \bar v} / \Pi'_h $ (see (\ref{phidef}), the prime 
denotes a derivative), with coefficients that are algebraic functions of 
the cross-ratio $\lambda$ of the four points defining the rectangle.

Theorem \ref{weakhol}  proves that all three crossing probabilities are 
weakly holomorphic second-order modular forms (definitions are given in 
Section \ref{hef}), and determines their 
leading Fourier terms at each cusp. Weak holomorphicity has been studied 
and led to striking results 
including the resolution of classical number-theoretical problems on 
partitions of integers and other results by
D. Zagier, R. Borcherds, K. Ono, K. Bringmann (\cite{B, Z, BO1, BO}, see 
also \cite{DJ} for further references). However we did not 
expect it in this context. 

Theorem  \ref{ThmPi} is a Hamburger-type 
result for the weakly holomorphic higher-order forms studied in this 
work. It proves that, under some technical assumptions analogous to those 
used in \cite{KZ}, all three crossing probabilities are determined by the 
transformation laws they satisfy.  Since the only physical input to the theorem
is  $\Pi'_h$, one finds that a single basic framework 
of assumptions leads to all three probability functions.  This points to an 
unknown connection between the functions $\Pi_{h}$, $\Pi_{h \bar v}$, and 
${\cal N}_h$.

This work is focused on the mathematical aspects of the new crossing 
probabilities.  We plan to explore its physical implications more fully 
elsewhere.  However, since the treatment here may also be of interest to 
physicists, we have included some comments that are intended to make the 
presentation more accessible to that audience.

\section{Crossing probabilities}
Set $T:=\left ( \smallmatrix 1 & 1 \\ 0 & 1 
\endsmallmatrix \right )$ and $S:=\left ( \smallmatrix 0 & -1 \\ 1 & 0 
\endsmallmatrix \right )$. We let $\G(2)$ denote the group of matrices of 
SL$_2(\mathbb Z)$ congruent to the identity mod$(2).$ We use the set of 
generators of $\G(2)$ consisting of
$g_1 := T^2$ and $g_2 := ST^{-2}S^{-1}=\left ( \smallmatrix 1 & 0 \\ 2 & 1 
\endsmallmatrix \right )$ (\cite{R}, p. 63).

Now set $q:=e^{2 \pi i z}$. Let $\eta$ be the Dedekind 
eta function $$\eta(z)=q^{1/24} \prod_{n=1}^{\infty}(1-q^{n}).$$
Also, let $\lambda$ denote the classical modular function for $\G(2)$.
It is given by 
\begin{equation}
\lambda(z):=16\frac{\eta(z/2)^8 \eta(2z)^{16}}{\eta(z)^{24}}=
1-\frac{\eta(z/2)^{16} \eta(2z)^8}{\eta(z)^{24}}.
\label{lambda}
\end{equation}
Neither the function $\lambda$ nor its derivative with respect to $z$
\begin{equation}
\lambda'(z)=16\pi i\frac{\eta(z/2)^{16} \eta(2z)^{16}}{\eta(z)^{28}}
\label{lambdapr}
\end{equation}
have any poles 
or zeros in $\HH$ (cf. \cite{R}, Section 7.2, (7.2.17)).

$\lambda(z)$ appears here because for $z=ir$ it is the cross-ratio of the four points 
to which the corners of a rectangle of aspect ratio $r$ are mapped.

Now let
\begin{equation} 
C := \frac{2^{1/3} \pi^2} {3 \, \Gamma(1/3)^3},
\end{equation} \label{Cdef}
following the notation in  \cite{KZ}.  Then
the horizontal crossing probability $\Pi_h(\lambda(ir))$  satisfies
\begin{equation} 
\frac{d}{dr} \left ( \Pi_h(\lambda(ir))\right ) = 
-4 \sqrt{3} \,C\, \eta(i r)^4,
\label{funct1}  
\end{equation}
as proven in \cite{Zi}. 
 
 The horizontal but not vertical crossing probability
  $\Pi_{h \bar v}(\lambda(ir))$ is proven in \cite{KZ} to satisfy
\begin{equation} 
\frac{d}{dr} \left ( \Pi_{h \bar 
v}(\lambda(ir))\right ) = 
-8\sqrt{3} f_2(ir),
\label{funct2} 
\end{equation} 
where 
\begin{eqnarray} \label{f2int}
f_2(z) &=& \frac{2 \pi i}{3} \eta(z)^4 \int_{\infty}^z 
f_3(w) dw,
\; \text{with} \\ \nonumber
f_3(w)&:=&\frac{\eta(w/2)^8 \eta(2w)^8}{\eta(w)^{12}}. 
\end{eqnarray}
The Dedekind  eta function is a weight $1/2$ 
cusp form on SL$_2(\mathbb Z)$ with a character. We denote the 
character of $\eta(z)^4$  by 
$\chi$. As shown in \cite{KZ}, $f_2(z)$ is a second-order modular form, of a kind which we
 specify  in the next section. 

Now let 
\begin{equation} \label{gdef}
g(z):=\eta(z/2)^{8} \eta(2z)^{8}/\eta(z)^{16}.
\end{equation}
This definition is useful for avoiding ambiguities from fractional
 powers of $\lambda(1-\lambda)$.
 Because of  
(\ref{lambda}) $g$ satisfies 
\begin{equation}
g^3=\frac{1}{16}\lambda(1-\lambda).
\label{g(z)}
\end{equation}

  Then, combining (\ref{lambda}) and (\ref{lambdapr}) gives, for $z\in 
\HH$, 
\begin{equation}
\eta(z)^4 = \frac{1}{16 \pi i}\frac{\lambda'(z)}{g(z)^2},
\label{eta4lambda}
\end{equation}
as well as
\begin{equation}
f_3(z) = \frac{1}{16 \pi i} \frac{\lambda'(z)}{g(z)}.
\label{f3lambda}
\end{equation}
Now let
\begin{equation}
\phi(z) :=\frac{f_2(z)}{\eta(z)^4}=\frac{1}{2}C\, 
\frac{\Pi'_{h \bar v}(z)}{\Pi'_h(z)}.
\label{phidef}
\end{equation}

Since $\Pi_{hv}(i/r) = \Pi_{hv}(ir)$ by symmetry, $\Pi'_{hv}(i) = 0$ (note that $r = 1$ is a square). Further, $\Pi_{hv} = \Pi_{h} - \Pi_{h \bar v}$, so that $\Pi'_{h}(i) = \Pi'_{h \bar v}(i)$.
Thus
 \begin{equation}
\phi(i) = \frac C{2}.
\label{phi(i)}
\end{equation}

Using (\ref{f2int}) to calculate $\phi'$ and (\ref{f3lambda}), one sees 
that $\phi$ is a function of  $\lambda$ only. Integrating gives 
\begin{equation}
\phi(z) = \frac1{2^{8/3}}  \, 
\lambda(z)^{2/3} \,{}_2F_1(1/3,2/3;5/3;\lambda(z)).
\label{phi}
\end{equation}

In a more recent work (\cite{SKZ}), the three new crossing probabilities 
mentioned in the Introduction are introduced and computed. For 
$0<\alpha<1$, and $1<\beta$, they are given by 
\begin{eqnarray}
\pi_h^b(\alpha, \beta)&=&\frac{(\beta+\alpha) \, _2F_1(1, 
4/3, 5/3, 1-\alpha/\beta)-2 \beta}{4 \sqrt{3} \, \pi  \, \beta^2(\beta-\alpha)},
\label{form1} \\
\pi_h^{\bar b}(\alpha, \beta)&=&\frac{(\beta+\alpha) \, 
_2F_1(1, 
4/3, 5/3, \alpha/\beta)+2\beta}{4 \sqrt{3} \, \pi \, \beta^2(\beta-\alpha)}
\label{form2}, {\rm and} \\ 
\nu_h(\alpha, \beta)&=&\frac{\beta^2+2 \alpha \beta-
(\beta^2-\alpha^2) \, _2F_1(1, 
4/3, 5/3, \alpha/\beta)}{4 \sqrt{3} \, \pi \, \beta^2(\beta-\alpha)^2}
\label{form3}. 
\end{eqnarray} 

In \cite{SKZ}  it is shown that a double integration of each of these 
quantities gives the crossing probabilities studied in \cite{KZ} as well as ${\cal N}_h$. 
Specifically 
\begin{eqnarray} 
\frac{1}{2}\Pi_{h\bar v}(\lambda)
&=&\int_0^{\lambda} \int_1^{\infty}\pi_h^{\bar b}(\alpha, \beta) \,d\beta \,
d\alpha 
\label{form11}, \\ 
\Pi_h(\lambda)-\frac{1}{2}\Pi_{h \bar v}(\lambda),
&=&\int_0^{\lambda} \int_1^{\infty}\pi_h^b(\alpha, \beta) \,d\beta \,d\alpha 
\label{form21}, \; {\rm and} \\ 
-\frac{1}{2}\Pi_{h\bar v}(\lambda)
+\frac{\sqrt{3}}{4 \pi} \log \left ( \frac{1}{1-\lambda}\right )
&=&\int_0^{\lambda} \int_1^{\infty}\nu_h(\alpha, \beta) 
\,d\beta \,d\alpha .
\label{form31} 
\end{eqnarray}
(The lhs of (\ref{form31}) is equal to ${\cal 
N}_h(\lambda)-\Pi_h(\lambda)$).
 
For clarity, we now introduce  the notation 
\begin{eqnarray} \label{newdef}
p_{\bar b}(z) &:=& \pi_h^{\bar b} (\lambda(z), 1) \; , \nonumber \\
 p_{b}(z) &:=&  \pi_h^{b} (\lambda(z), 1) \; ,\; {\rm and}  \nonumber  \\
n(z)   &:=& \nu_h(\lambda(z), 1) \; ,
\end{eqnarray}
 for the holomorphic functions of  $z \in \HH$ obtained
from (\ref{form1}), (\ref{form2}), and (\ref{form3}) after we replace $\alpha$ and $\beta$ with $\lambda(z)$  and
$1$ respectively. 

We next prove (with prime standing for differentiation in $z$)
\begin{thm} \label{1stthm} For all  $z \in \HH$,
\begin{eqnarray}
\lambda'(z) \, p_{\bar b}(z)&=&
4 \sqrt {3} \, i \left (
\frac{\lambda(z) f_2(z)}{\lambda'(z)} \right )'
\label{eq1} \; , \\
\lambda'(z) \, p_{b}(z)&=&
4 \sqrt {3}\,C \, i \left ( 
\frac{\lambda(z)}{\lambda'(z)} \eta^4 \right )'-
4 \sqrt {3} \, i \left ( \frac{\lambda(z) f_2(z)}{\lambda'(z)} \right )'
\label{eq2}\; , \; {\rm and} \nonumber \\ \\
\lambda'(z) \, n(z) &=&
\frac{\sqrt{3} }{4 \pi} 
\left ( \frac{\lambda(z)}{1-\lambda(z)} \right )'
-4 \sqrt {3} \, i \left ( \frac{\lambda(z) f_2(z)}{\lambda'(z)} \right )' \; .
\label{eq3}
\end{eqnarray}
\end{thm} 
\begin{proof}
We first deduce from (\ref{form2}) that 
$\pi_h^{\bar b}(\alpha, \beta)=\beta^{-2} \, \pi_h^{\bar b}(\alpha/\beta, 1).$
This and the change of variables $ \alpha/\beta \to \beta$ 
imply that the rhs of (\ref{form11}) equals
$$\int_0^{\lambda} \int_0^{\alpha} \pi_h^{\bar b}(x, 1)dx \frac{d 
\alpha}{\alpha}$$
and hence a differentiation in terms of $i r$ 
together with (\ref{funct2}) (recalling that 
the differentiation in (\ref{funct2}) is in terms of
$r$ and not $ir$) and re-arranging gives
$$ 
4 i \sqrt{3} \,  \frac{ \lambda(ir) f_2(ir)}{\lambda'(ir)}= 
\int_0^{\lambda(ir)} \pi_h^{\bar b}(x, 1)dx. 
$$
One more differentiation in terms of $ir$ and the definitions (\ref{newdef}) implies the result for $z=ir$.
The analytic continuation of $\pi_h^{\bar b}(\cdot, 1)$ and the 
fact that $\lambda$, $\lambda'$ do not have any poles or zeros in $\HH$ 
imply that it holds for all $z \in \HH$.

The assertions about the other two crossing probabilities are proved in exactly 
the same way using (\ref{form21}), (\ref{form31}).
\end{proof}

We next re-write $p_{\bar b}$, $p_{b}$, and $n$ in a way that avoids derivatives of 
$\lambda$ 
\begin{prop}  \label{PropDZ} For every $z \in \HH$, we have 
\begin{eqnarray*}
p_{\bar b}(z) &=&
\frac{1}{4 \sqrt{3} \, \pi \,g(z)^2} \,  \frac{1+\lambda(z)}{1-\lambda(z)} \,  \phi(z)+
\frac{1}{2 \sqrt{3} \, \pi} \frac{1}{1-\lambda(z)} \\
p_{b}(z) &=&
\frac{1}{4 \sqrt{3}  \,\pi \,g(z)^2} \, \frac{1+\lambda(z)}{1-\lambda(z)} 
\left (2\,\phi(i)-\phi(z) \right )
-\frac{1}{2 \sqrt{3} \, \pi} \frac{1}{1-\lambda(z)} \\
n(z) &=&
-\frac{1}{4 \sqrt{3} \, \pi\, g(z)^2} \,  \frac{1+\lambda(z)}{1-\lambda(z)}  \, 
\phi(z)+
\frac{1}{4 \sqrt{3}  \,\pi} \,  \frac{1+2 \lambda(z)}{(1-\lambda(z))^2}.
\end{eqnarray*}
\end{prop} 
\begin{proof} 
By (\ref{f3lambda}),
\begin{equation} \lambda'(z)=16 \pi i  \, g(z)^2\, \eta(z)^4.
\label{lprime}
\end{equation}
Using (\ref{g(z)}) and (\ref{lprime}),
$$3g'g^2=\frac{1}{16}(1-2\lambda)\lambda'=
\pi i (1-2 \lambda ) \, g^2  \, \eta^4.$$
This and (\ref{lprime}) imply that
$$
\left ( \frac{\lambda}{\lambda'} \eta^4\right )'=
\left ( \frac{\lambda}{16 \pi i \,  g^2} \right )'=
\frac{1}{16 \pi i} \frac{\lambda' g - 2g' \lambda}{g^3}=
\frac{\eta^4(1+\lambda)}{3(1-\lambda)}
$$
 Expanding (\ref{eq1}) according to this formula, from (\ref{f3lambda}) 
and the definition of  $f_2$ we deduce 
$$p_{\bar b}(z) =
\frac{1}{4 \sqrt{3}  \, \pi g(z)^2} \, 
\frac{1+\lambda(z)}{1-\lambda(z)} \,  \phi(z)+
\frac{1}{2 \sqrt{3} \pi } \frac{1}{1-\lambda(z)}.
$$

The proof of the  formula for $p_{ b}$ is similar. The formula for $n$ may be obtained by comparing (\ref{form2}) with (\ref{form3}) and making use of the formula for $p_{\bar b}$.
\end{proof}

\begin{rem} \rm  
 Using (\ref{g(z)}) to write $g$ in terms of  $\lambda$,  Proposition \ref{PropDZ} gives each of  $p_{\bar b}$, $p_{b}$, and $n$ as an expression linear in the ratio $\phi$ (see (\ref{phidef})), with the coefficients being  rational functions of $\lambda$.  Thus all three probabilities are determined by the same physical quantity $\phi$, which is itself proportional to $\Pi'_{h \bar v} / \Pi'_h$.  Note also that the constant $C = 2\, \phi(i)$ in $p_{b}$ is exactly the constant $C$ in (39) of \cite{KZ}. It follows that $2\,\phi(i) - \phi(z)$ and
 $\phi(z)$ transform into each other under $S$.
\end{rem}
\begin{rem} \rm  Proposition \ref{PropDZ} illustrates why it is natural 
to work with the group $\G(2)$ rather than the {\it theta group} 
$\G_{\theta}$, as in \cite{KZ}.  The ratio $\phi$ is easily seen--use, 
e.g.\  (39)  of \cite{KZ}--to transform as a 
second-order modular form (for the definition see \cite{KZ} or the next section) under the generators of either $\G(2)$ (e.g.\ $\{g_1, g_2\}$) or $\G_{\theta}$ (e.g.\ $\{g_1, S\}$).  However, $\lambda|_0 \,S = 1-\lambda$; as a result the rational functions of $\lambda$ of Remark 2.3 do not transform simply under $\G_{\theta}$.
\end{rem}  
 
\section{Higher-order forms} \label{hef}

Let $\G$ be a congruence subgroup  of SL$_2(\mathbb Z)$ and $k 
\in 2\mathbb Z$. As usual, for every character $\chi$ on $\G$
we define an action of 
$\G$ on the space of functions $f: \HH \to \C$ given, for $\g \in \G$, by 
$$(f|_{k, \chi} \gamma)(z)=f(\gamma z) \,  j(\gamma, z)^{-k}  \, 
\overline{\chi(\gamma)}.$$
Here, $j(\left ( \smallmatrix * & * \\ c & d \endsmallmatrix \right 
)):=cz+d$. When $\chi$ is the trivial character $\mathbf 1$ we write $f|_k  \, 
\gamma$. 

To describe the condition on the ``growth at the cusps" which will be 
included in the definition of higher-order forms we first note that a \it 
cusp \rm is a point $x \in \mathbb R \cup \{ \infty \}$ such that $\g x=x$ for 
some $\pm 1 \ne \g \in \G$ with $|$tr$(\g)|=2$. Cusps are important 
because they are the only points on the real line at which modular forms 
have a regular behavior. Indeed, one of their special features is that if 
a function is $\G$-invariant under the action $|_k$, then 
its composition by an appropriate map sending infinity to a cusp will be 
periodic and hence it will possess a Fourier expansion. This fact is 
crucial in establishing 
the finite
dimensionality of the space of modular forms. As a result (as will be 
seen in the proof of Theorem \ref{ThmPhi}) in order to uniquely determine 
functions, in addition to the transformation laws satisfied, detailed knowledge 
of the growth at the cusps is essential. 

Two cusps $\ca, \cb$ are 
called \it 
equivalent \rm if there is a $\g \in \G$ such that $\ca =\g \cb$ and 
\it inequivalent \rm otherwise. 

For each cusp $\ba$ there is a \it 
scaling matrix, \rm i.e. a matrix $\sa$ such that 
$$\sa(\infty)=\ba \qquad \text{and} \, \, \, \, 
\sigma^{-1}_{\ba}\Gamma_{\ba}\, \sigma_{\ba}=\Gamma_{\infty}$$ 
where $\Gamma_{\ba}$ (resp. $\Gamma_{\infty}$) is the set of elements of 
$\G$ fixing $\ba$ (resp. $\infty$) (\cite{DKMO}, Section 2).
If $f$ is a function on $\HH$ such that $f|_k \, \sa$ has a Fourier 
expansion of the form 
$$(f|_k \, \sa)(z)=\sum_{n=-m}^{\infty} a_n\, e^{\pi i n z/a}$$
for some $m \in \mathbb Z_{\ge 0}$ with $a \in \mathbb Q_+$ and $a_{-m} 
\ne 0$, then
we say that $f$ is \it meromorphic \rm at $\ca$. If $m=0$ (resp. $m<0$), 
then $f$ is \it holomorphic \rm (resp. {\it cuspidal}) at $\ca$.

For a set $\{\chi_1, \dots, \chi_n\}$ of characters on $\Gamma$, we 
define a \it weakly holomorphic nth-order modular form on $\Gamma$ of 
weight $k$ and 
type $(\chi_1, \dots, \chi_n)$ \rm to be a holomorphic function on $\HH$
which is meromorphic at the cusps and 
such that, for every $\gamma_i \in \Gamma$, 
$$f|_{k, \chi_1}(\gamma_1-1)|_{k, \chi_2}(\gamma_2-1) 
\dots|_{k, \chi_n}(\gamma_n-1)=0.
$$
If $f$ is holomorphic (resp. cuspidal) at all cusps, then we call the form 
holomorphic (resp. cuspidal). We denote the space of order $n$ weakly 
holomorphic (resp. holomorphic, cuspidal) forms of weight $k$ and type
$(\chi_1, \dots, \chi_n)$ on $\G$ by
$\tilde M_k^n(\G; \chi_1, \dots, \chi_n)$
(resp. $M_k^n(\G; \chi_1, \dots, \chi_n)$,
$S_k^n(\G; \chi_1, \dots, \chi_n)$). 

 Note that, in contrast to the case of ordinary modular forms, a Fourier 
expansion at the cusps is not guaranteed. The reason is that this would 
necessitate that $f$ be invariant under a 
subgroup of $\sa \Gamma_{\infty} \sa^{-1}$ which, in general, does not 
happen when the order is higher than $1$.

We next prove that the three new crossing probabilities $p_{\bar b}$, $p_{b}$, and $n$ are 
second-order forms and  determine their growth at the cusps. 

The group $\G(2)$ has three inequivalent cusps at 
$\infty$, $0$ and $-1$. Three corresponding scaling matrices are $I$, 
$U:= \left ( \smallmatrix 0 & -1 \\ 1 & 1 \endsmallmatrix \right ) = S T$
and
$U^2=\left ( \smallmatrix -1 & -1 \\ 1 & 0 \endsmallmatrix \right )$.

We can then prove 
\begin{thm} As functions of $z$,  
$p_{\bar b}(z)$, $p_{ b}(z)$ and
$n(z)$ are weakly holomorphic second-order modular forms 
on 
$\Gamma(2)$ of weight $0$ and type $(\mathbf 1, \chi)$. The first power 
of $q$ appearing in the expansion of $p_{\bar b}$ at $\infty$ (resp. 
$0$, $-1$) is $1$ (resp. $q^{-5/6}$, $q^{2/3}$). The corresponding first 
powers for $p_{ b}$ and $n$ are $q^{-1/3}$, $1$, $q^{2/3}$
and $q^{1/2}$, $q^{-1}$, $q^{2/3}$.
\label{weakhol}
\end{thm}
\begin{proof}
The function $f_2$ is a holomorphic weight $2$ second-order form 
of type $(\mathbf 1, \chi)$ satisfying 
\begin{equation}
f_2|_2(\gamma-1)=d_{\gamma}  \, \eta^4
 \qquad \text \, \, \text{for all} \, \,  \gamma \in \Gamma_{\theta}
\label{f.e.}
\end{equation}
for a constant $d_{\gamma}$ depending solely on $\gamma$ (\cite{KZ}, 
(38)). Since $\Gamma(2)$ is a subgroup of $\Gamma_{\theta}$, (\ref{f.e.}) 
holds for $\Gamma(2)$ too. It easy to verify that $d_{g_1}=0$ and 
$d_{g_{2}}=C\,(e^{-2 \pi i/3}-1)$, so the rhs of (\ref{f.e.}) is not zero for all $\gamma \in \G(2)$.

An easy computation (see also Theorem 2.2 of \cite{CD}) 
implies that, for $\gamma \in \Gamma(2)$,
$$
\left (\frac{\lambda}{\lambda'}f_2 \right )|_0  \, \gamma=
\left (\frac{\lambda}{\lambda'} \right )|_{-2} \,  \gamma \,  f_2 |_2  \, 
\gamma=\frac{\lambda}{\lambda'}f_2 +
\frac{d_{\gamma} \lambda \eta^4}{\lambda'}.
$$
Therefore 
\begin{equation}
\left (\frac{\lambda}{\lambda'}f_2 \right )'|_2 (\gamma-1)=
\left ( \frac{d_{\gamma} \lambda \eta^4}{\lambda'} \right )'.
\label{weak}
\end{equation}
which, by Theorem \ref{1stthm}, implies
\begin{equation}
p_{\bar b}|_0  \, (\gamma-1)=
d_{\gamma}  \, G
\label{weak1}
\end{equation}
where 
$$G:=\frac{1}{\lambda'}\left ( \frac{\lambda \eta^4}{\lambda'} 
\right )'.
$$
Since $\lambda \eta^4/\lambda'$ transforms as a modular form of weight 
$0$ with character $\chi$ on $\Gamma(2)$, its derivative transforms as 
a weight $2$ form and, therefore, $G$ transforms as a weight $0$ form with 
character 
$\chi$.

The behaviour of $p_{b}$ under the action of $\G(2)$ follows from 
(\ref{eq2}), which implies that 
$p_{b}=-p_{\bar b}+4\sqrt{3}\,C \, i  \, G$. 
Since $G|_0(\gamma-1)=(\chi(\g)-1)G$, we conclude that 
$p_{b}|_0(\g-1)$ transforms as a weight $0$ modular form with character 
$\chi$.

In the same way, since $(\lambda/(1-\lambda))'/\lambda'$ transforms as a 
weight $0$ modular form with trivial character, we deduce that 
$n|_0(\g-1)$ transforms as a weight $0$ modular form with character 
$\chi$.

To complete the proof  that $p_{\bar b}(z)$, 
$p_{ b}(z)$ and $n(z)$ are second-order 
weakly holomorphic forms of weight $0$ and type $(\mathbf 1, \chi)$ for 
$\Gamma(2)$, it remains to show that $p_{\bar b}(z)$ and 
$G$ are meromorphic at the cusps. The
argument also proves the last 
part of the theorem, specifying the first terms of the Fourier expansions 
of the three probability functions.

We start by determining the first powers of $q$ appearing in the Fourier 
expansions of $\lambda$ at the three inequivalent cusps $\infty$, $0$ and 
$-1$. We shall be making use of  (\ref{eta4lambda}), (\ref{phi}), 
and the transformations of $\lambda$: $\lambda|_0\,S = 1-\lambda$ and 
$\lambda|_0\,T = \lambda/(\lambda-1)$ (\cite{R}, (7.2.2)).  Thus 
\begin{equation}
\lambda|_0\,U = 1/(1-\lambda) \qquad  \text{and} \, \, 
\lambda|_0\,U^2 = 1-1/\lambda.
\label{lambdatransf}
\end{equation}
Since the leading term in $\lambda$
at $\infty$ is $16 q^{1/2}$, it follows that the leading terms at $0$ and 
$-1$ are $1+16q^{1/2}$ and $q^{-1/2}$ respectively. 
Therefore $\lambda'$ is a weight $2$ first-order weakly 
holomorphic form on $\G(2)$ and the first powers in its 
Fourier expansions at the cusps 
$\infty$, $0$ and $-1$ are  $q^{1/2}, q^{1/2}$ and 
$q^{-1/2}$ respectively. 
The corresponding powers for $\lambda/\lambda'$ 
are then $1$, $q^{-1/2}$ and $1$.  

Next we verify that $p_{\bar b}(z)$, $p_{ b}(z)$ and $n(z)$ 
have Fourier expansions at $\infty, 0, -1$. 
We first note that the Fourier expansion of $\lambda$ and 
(\ref{lambdatransf}) imply that $\lambda$ maps the cusps $\infty$, $0$ and 
$-1$ to $0$, $1$, and $\infty$ respectively. On the other hand,
each of the probabilities $\pi_h^{\bar b}(\lambda(z), 1)$, 
$\pi_h^b(\lambda(z), 1)$ and $\nu_h(\lambda(z), 1)$ is given (see 
(\ref{form1}), (\ref{form2}), and (\ref{form3})) in terms of a 
hypergeometric function of $\lambda$ or $1-\lambda$ and integral powers 
of $\lambda$ only. It follows from the classical linear transformation 
formulas for hypergeometric 
functions that  an expansion around any of the cusps will involve 
integral powers of  $\lambda$ or $\lambda^{1/3}$ only.  Thus any of these 
expansions is invariant under $T^6$, which establishes the 
periodicity of $g$, $g|_0U$, $g|_0U^2$, for $g=\pi_h^{\bar b}(\lambda(z), 
1), \pi_h^b(\lambda(z), 1)$ and $\nu_h(\lambda(z), 1)$.
 
The leading term in each of  $\pi_h^{\bar b}(\lambda,1)$, 
$\pi_h^b(\lambda,1)$ and $\nu_h(\lambda,1)$ and therefore $p_{\bar b}(z)$, 
$p_{ b}(z)$ and $n(z)$
 at each cusp then follows 
immediately from (\ref{form1}), (\ref{form2}) and (\ref{form3}) using  
the Taylor expansion of the hypergeometric function and  $\lambda 
\sim q^{1/2}$.

The first terms of the Fourier expansions of $G$ follow in similar 
(but simpler) way from the Fourier expansions of $\lambda$, $\lambda'$
and the observation that the first power in the expansion of $\eta^4$ is
$q^{1/6}$. 

Finally, one must also verify the meromorphicity at all equivalent cusps.
For $G$, this is automatically implied by its invariance under the action 
of $\G(2)$. For the three probability functions, it is no longer automatic 
because they are not $\G(2)$-invariant. Let $\bb=\g \,
\ba$ ($\g \in 
\G(2)$) for $\ba= \infty$, $0$ or $-1$. We have 
$\sb =\g \,\sa$, $\G_{\sb}=\g\, \G_{\sa} \g^{-1}$ and that if $\sb'$
is another scaling matrix for $\bb$ then $\sb=\sb' \,T^m$, with $m \in 2 
\mathbb Z_{\ge 0}$ (\cite{DKMO}, Section 2). Therefore, if 
$\sb'$ is 
any scaling matrix of $\bb$, we have (with (\ref{weak1})) 
$$p_{\bar b}|_0\, \sb'=p_{\bar b}|_0\,\sb \,T^m=
p_{\bar b}|_0\,\g \,\sa \,T^m=
p_{\bar b}|_0\, \sa\, T^m+c_{\g}\, G|_0 \,\sa \,T^m
$$
From the behaviour of $p_{\bar b}$ at $\infty, 0, -1$ we proved above 
as well as the behaviour of $G$ at these cusps, we deduce the 
meromorphicity of $p_{\bar b}$ at all cusps $\bb$. 
The same argument applies to $p_{b}$ and $n$.

For convenience, Table $1$ lists the first terms of the Fourier expansions 
for the functions  discussed.

\begin{table}
\label{Table}
\centering
\caption{First terms of Fourier expansions} \nonumber
\begin{tabular}{|l|l|l|l|} 
\hline
\hline
 & $\infty$ \rule {0pt}{2.6ex}& $0$ & $-1$ \\[2 pt] \hline
$\lambda$  \rule {0pt}{2.6ex}& $q^{1/2}$ & $1+16q^{1/2}$ & $q^{-1/2}$ \\[4 pt] \hline
$\frac{\lambda}{\lambda'} \rule {0pt}{2.6ex}$ & $1$ & $q^{-1/2}$ & $1$\\[4 pt] \hline
$\lambda'$ & $q^{1/2} \rule {0pt}{2.6ex}$ & $q^{1/2}$ & $q^{-1/2}$ \\[4 pt] \hline
$\frac{1}{\lambda'} \rule {0pt}{2.6ex}$ & $q^{-1/2}$ & $q^{-1/2}$ & $q^{1/2}$ \\[4 pt] \hline
$p_{\bar b}(z) \rule {0pt}{2.6ex}$ & $1$ & $q^{-5/6}$ & $q^{2/3}$ \\[4 pt] 
\hline
$p_{ b}(z) \rule {0pt}{2.6ex}$ & $q^{-1/3}$ & $1$ & 
$q^{2/3}$ \\[4 pt] 
\hline
$n(z) \rule {0pt}{2.6ex}$ & $q^{1/2}$ & $q^{-1}$ & 
$q^{2/3}$ \\[4 pt] 
\hline
$\left (\frac{\lambda}{\lambda'} \eta^4 \right )' \rule {0pt}{2.6ex}$ & $q^{1/6}$ & 
$q^{-1/3}$ & $q^{1/6}$ \\[4 pt] \hline
$\frac{1}{\lambda'} \left (\frac{\lambda}{\lambda'} \eta^4 \right )' \rule {0pt}{2.6ex}$ 
 & $q^{-1/3}$ & $q^{-5/6}$ & $q^{2/3}$ \\[4 pt] \hline  
\end{tabular} 
\end{table}

\end{proof}

\begin{rem} \rm
Theorem \ref{weakhol}, in combination with Prop. 2.2,
shows that $\phi$ is a weight $0$ weakly holomorphic second-order
modular form on $\Gamma(2)$ of type $(\mathbf 1, \bar \chi)$. This result is already implied in 
 \cite{KZ}, but here it is proved explicitly.
\end{rem} 
 
\section{Uniqueness theorems}

We call \it conformal block \rm 
of dimension $\alpha \in \mathbb R$ a function of $r > 0$ 
expressible in the form
\begin{equation} \label{blockdef}
\sum_{n=0}^{\infty}a_n \, e^{-\pi (n+\alpha) r},
\end{equation}
with $a_{0} \ne 0$. As noted in \cite{KZ}, if $\Pi(r)$ is a conformal 
block, then the convergence of (\ref{blockdef}) implies that the function 
 \begin{equation}
 P(ir) := \Pi(r),
 \end{equation}
 extends to a holomorphic function $P(z)$ with $z \in \HH$ (its power series expansion is exactly as in (\ref{blockdef}), except that $-r$ is replaced by $i z$). In \cite{KZ}, transformation properties under $S$, which maps the imaginary axis into itself, were important.  Here, the corresponding group element is $g_2: z \to z/(1+2z)$, which does not have this property.  Therefore we use a slightly different approach,  working exclusively with the  analytic continuations of the conformal blocks.
 
We start with a Lemma which will be needed in the proof of the next theorem.
 \begin{lem} \label{lemS4} A function 
$f \in S_4(\G(2), \chi)$ is uniquely determined by 
the coefficients of $e^{\pi i z/3}$ in its
Fourier expansions at $\infty$ and $0$. 
\end{lem}
\it Proof of Lemma: \rm
For each cusp $\ca$, let
$$E_{\ca}(z):=2i \lim_{s \to 1} \frac{d}{dz}
\left ( \sum_{\g \in \G_{\ca} \backslash \G} \text{Im}
(\sigma_{\ca} \g z)^s \right ).$$
Then
$$E_{\ca}(z)-E_{\infty}(z)=-1-4 \pi^2 \sum_{n=1}^{\infty}n \left (Z_{\ca,
\infty}(n, 0, 1)-Z_{\infty, \infty}(n, 0, 1)\right ) e^{2 \pi i n z},$$
where $Z_{\ca, \cb}(m, n, s)$ is the Selberg-Kloosterman
zeta function. (Its definition is given in \cite{GO}, but is rather lengthy. 
Since we do not use it further, we refrain from quoting it here.) 
The set
$\{E_{-1}-E_{\infty}, E_{0}-E_{\infty}\}$ spans the direct complement
of the space of cusp forms within the space of all modular forms of
weight $2$ for
$\G(2)$ (\cite{GO}).  Further, the dimension of $S_4(\G(2), \chi)$ is
$2$ (\cite{HK}, Section 4.). Therefore a basis is given by
$(E_0-E_{\infty}) \, \eta^4$ and $(E_{-1}-E_{\infty}) \, \eta^4$. As a
result, any function
$f \in S_4(\G(2), \chi)$ has coefficients of $e^{\pi i z/3}$ in its
Fourier expansion at $\infty$ and $0$, which we denote by $a$ and $b$,
respectively.

Thus
$$f=x_1(E_{0}-E_{\infty}) \, \eta^4+x_2 (E_{-1}-E_{\infty}) \, \eta^4,$$
for some $x_1, x_2 \in \C$.
Further (see \cite{GO}), for any cusps $\ca, \cb, \cc$,
$$j(\sc,z)^{-2}(E_{\ca}(\sc z)-E_{\cb}(\sc z))
=\delta_{\ca \cc}-\delta_{\cb \cc}+O(e^{-2\pi y}),
$$
where $z = x+i y$, which gives
$$ f=-(x_1+x_2)e^{\pi i z/3}+O(e^{-7 \pi y/3})\qquad
f|_4 \, U=x_1 e^{\pi i z/3}+O(e^{-7 \pi y/3}). $$
Thus $x_1=b$ and $x_2=-a-b.$ \qed \\

We are now in a position to prove 
\begin{thm} \label{ThmPhi} Let 
$F(z)=\sum_{n=0}^{\infty}b_n \, e^{\pi  i(n+\frac{1}{3}) z}$
be the analytic continuation of a conformal block of dimension 
$1/3$,  with $b_0 = \pi i/3$. 
Suppose that 
\begin{enumerate}
\item[(a)] 
\begin{equation} \label{F4a}
F|_4 \, g_2 =\chi(g_2) F,
\end{equation}
along some curve in $\HH$,
\item[(b)]
$F(-1+i/r)$ is 
bounded as $r \to \infty$, and 
\item[(c)]
\begin{equation} 
\lim_{r \to \infty} e^{\frac{\pi r}{3}}r^{-4}F \left (\frac{i}{r} 
\right )= -\frac{4}{3} \, 2^{1/3} \pi^2. 
\label{limitsa}
\end{equation}
\end{enumerate}
Then $F(z) \in S_4(\G(2), \chi)$ and
$$F(z)=\frac{\lambda'(z)}{\lambda(z)} \left ( 
\frac{\lambda(z) \, \eta(z)^4}{\lambda'(z)}\right )'.$$
\end{thm}
\begin{proof} 
We first prove that $F$, if it exists, is a weight $4$ cusp form for $\G 
(2)$ with character $\chi$.

Since $F$ is analytic (\ref{F4a}) holds on all of 
$\HH$.
The Fourier expansion of $F$ and $\chi (T^2)
=e^{2 \pi i/3}$ then imply 
$$F|_{4, \chi} \, T^2 =F.$$ 
Therefore, $F|_{4, \chi} \, \g=F,$ for all $\g \in \G (2)$.

We next establish the vanishing at the cusps. First, by 
its Fourier expansion, $F$ vanishes at $\infty$. To verify that the 
Fourier expansion of $F|_4 \, U$ has only positive powers
(and thus that $F$ vanishes at $0$), it suffices to show 
$(F|_4 \, U)(-1+ir) \to 0$ as $r \to \infty$. But
$(F|_4 \, U)(-1+ir)= F(i/r) r^{-4}$, which  must vanish  as $r \to \infty$  by (\ref{limitsa}). For the cusp at $-1$ consider $(F|_4 \, U^2)(ir)= F(-1+i/r) \, r^{-4}$.  By the assumed boundedness of $F(-1+i/r)$, this vanishes as $r \to \infty$.

By  assumption, the 
coefficient of $e^{\pi i z/3}$ in the expansion of $F$ at $\infty$ is 
$$b_0 = \lim_{z \to \infty} e^{-\pi i z/3} F(z)=\frac{\pi \, i}{3}.
$$
The corresponding coefficient at $0$ is $\lim_{z \to \infty} e^{-\pi i 
z/3} (F|_4 \, U)(z)$.  Taking the limit over the line $z=-1+ir,$ as $r \to 
\infty$, and using  (\ref{limitsa}), one has 
$$\lim_{r \to \infty} F(i/r) r^{-4} e^{\pi r/3}e^{\pi i 
/3}=-\frac{4}{3} \, 2^{1/3} e^{\pi i/3} \pi^2.$$

A computation confirms that these coincide with the corresponding 
coefficients
for $\frac{\lambda'}{\lambda} (\frac{\lambda}{\lambda'} \eta^4)'.$ Lemma 
\ref{lemS4} then implies the result.  
\end{proof}

\begin{rem} \rm Theorem \ref{ThmPhi} resembles Theorem $1$ of \cite{KZ}.  
In that case, the assumption of modular transformation properties under $S$ led to a proof that an even conformal block (i.e.\ one with all $a_{2n+1} = 0$) has dimension $1/3$ and is in fact equal to Cardy's function $\Pi_h$.  Here we assume the block dimension and modular transformation property under the generator $g_2$ of $\G(2)$ and find that the function is a simple expression involving $\eta^4$, which is proportional to the $z$-derivative of $\Pi_h$.  The assumption of evenness is not necessary.
\end{rem}
\begin{rem} \rm
Note that if we let  $z=ir/(1+2ir)$ with $r>0$ be the curve in 
(\ref{F4a}), the assumption of transformation under $g_2$, then the lhs 
of the equation is in the physical regime, 
i.e.\ $z = ir$.
\end{rem}
.
\begin{thm} \label{ThmPi} Let
$P(z)=
\sum_{n=0}^{\infty}a_n \, e^{\pi i (n+1) z}$
be the analytic continuation of a conformal block of dimension 
$1$,  and let $F(z)$ be as in Theorem \ref{ThmPhi}. 
Set $$P_1:=4 \sqrt{3} \,C \, i \, F - P.
$$
Further, for a fixed $A \in \mathbb C$   set
 $$\tilde P(z):=P(z)+A \, F(z).$$
Suppose that
\begin{enumerate}
\item[(a)] \begin{equation}
\tilde P|_4 \, g_2
=\tilde P,
\label{fe2}
\end{equation}
along a curve in $\HH$ and
\item[(b)] $P(-1+i/r)$ is bounded as $r \to \infty$.
\end{enumerate}
Then 
\begin{enumerate}
\item[(i)] If $A=-C$ and 
$P(i/r)$ is bounded as $r \to \infty$,
$$\qquad P(z)=\frac{(\lambda'(z))^2}{\lambda(z)} p_{\bar 
b}(z),$$
\item[(ii)] Under the same assumptions,
$$P_1(z)=\frac{(\lambda'(z))^2}{\lambda(z)} p_{b}(z).$$
\item[(iii)] If $A=C$ and 
$P(i/r)r^{-4}  \to - \sqrt 3 \pi /4$ as $r \to \infty$,  
$$P(z)=\frac{(\lambda'(z))^2}{\lambda(z)} n(z).$$
\end{enumerate}
\end{thm}
\begin{proof} We consider the claims in order. \\

\it Proof of (i) \rm 
We proceed exactly as in the proof of Theorem \ref{ThmPhi}, except replacing $\chi$ by $1$.  Since $\tilde P$ (and $P$) are analytic,  if they exist,
(\ref{fe2}) holds in all of $\HH$.
Since $g_2 = ST^{-2}S^{-1}$, (\ref{F4a}) and $\chi (T^2)
=e^{2 \pi i/3}$ give
\begin{equation}
P|_4 \, (g_2 -1)=A(1-e^{-2 \pi i /3}) F \in S_4(\G(2), \chi).
\label{t.law0}
\end{equation}
Further, the Fourier expansion of $P$ implies that $P|_4 \left ( 
T^2-1 \right ) =0.$
Hence
\begin{equation}
P|_4 (\g-1) \in S_4(\G(2), \chi).
\label{t.law}
\end{equation}
As in the proof of Theorem \ref{ThmPhi}, we have by assumption that $P(ir)$, $(P|_4 \, U)(-1+ir)$,  and 
$(P|_4 \, U^2)(ir) \to 0$ as $r \to \infty$, establishing that $P(z)$ vanishes at all three cusps. By (\ref{t.law}),  $P|_4 \, \g$ ($\g \in \G(2)$) also vanishes at all three cusps. 

Now let 
$$L:=\frac{ (\lambda')^2} {\lambda} p_{\bar b}.$$
 Then
\begin{eqnarray} \nonumber
(P-L)|g_1 &=& P|g_1-L|g_1 \\
&=& P-L,
\end{eqnarray}
by the periodicity of both $P$ and $L$.  Further
$$(P-L)|_4(g_2-1)=A(1 - e^{-2\pi i /3})F- \frac{ (\lambda')^2} {\lambda} p_{\bar b}|_4(g_2-1),$$
by (\ref{t.law0}). By (\ref{weak1}) and Theorem \ref{ThmPhi},
$$p_{\bar b}|_4(g_2-1)=d_{g_2}G=d_{g_2} \frac{\lambda}{(\lambda')^2} F.$$
 Therefore, since $(\lambda')^2/ \lambda$ is invariant under $g_2$,
\begin{eqnarray} \nonumber
(P-L)|_4(g_2-1) &=& A(1 - e^{-2\pi i /3})F-d_{g_2}F \\
&=& 0,
\end{eqnarray}
using the formula for $d_{g_2}$ after (\ref{f.e.}).  Thus $P-L$ is invariant under both generators of $\G(2)$.  
As noted when recalling the definition of cusps, this $\G(2)$-invariance 
in terms of the action $|_4$ also implies that $L$ has Fourier 
expansions at each cusp of  $\G(2)$. Theorem \ref{weakhol} (see Table  
$1$)  implies that it vanishes at each cusp.  Since $P$ vanishes as well, 
so does $P-L$, and the latter is a standard weight $4$ cusp form with 
trivial character for $\G(2)$.  Therefore it vanishes (\cite{HK}), 
establishing the uniqueness of $P$ (and $\tilde P$). This 
completes the proof of $(i)$. 

\it Proof of (ii): \rm This is immediate from $(i)$, Theorem 
\ref{ThmPhi}, and (\ref{eq2}).

\it Proof of (iii): \rm Analogous to the proof of $(i)$, except that $L$ 
is redefined with $n$ replacing $p_{\bar b}$.  Here, 
neither $P$ nor $L$ vanishes at the cusp at $0$, but $P-L$ does, which is 
sufficient.

\end{proof}
\begin{rem} \rm Theorem \ref{ThmPi} resembles Theorem $3$ of \cite{KZ}.  In that case, the assumption of modular transformation properties under $S$ for two conformal blocks of arbitrary dimension led to  a characterization of  crossing probabilities that generalize $\Pi_h$ and $\Pi_{h \bar v}$.  Here, by assuming the block dimensions, modular transformation properties under the generator $g_2$ of $\G(2)$, and cusp properties, we reproduce all three new crossing probabilities.  
\end{rem}
\begin{rem} \rm It is interesting that the only input to Theorem \ref{ThmPi} from the physics of the problem is via $F$. The physical input to $F$ comes from $\eta^4$, itself proportional to the $z$-derivative of $\Pi_h$. Thus, under the assumptions
 of the theorem, all three crossing probabilities are determined by $\Pi_h$.
\end{rem}

\section{Acknowledgements}
ND thanks Marvin Knopp for useful conversations on certain 
aspects of the paper. PK is grateful to Don Zagier for leading us 
to Proposition \ref{PropDZ}, 
for other useful conversations, and for his kind hospitality at the 
Coll{\`e}ge de France, where part of this research was done.  He also thanks 
C. Hongler for an advance copy of \cite{HS} and J. J. H. Simmons and R. M. Ziff for useful
conversations.

This work was supported in part by the National Science Foundation Grant 
No. DMR-0536927 (PK).


\begin{thebibliography}{5}

\bibitem{B} 
R. Borcherds, \textsl{\it Automorphic forms on O$_{s+2,2}(\mathbb R)$ and 
infinite products}, Invent. Math. 120 (1995), no. 1, 161--213.


\bibitem{BO1} 
K. Bringmann and K. Ono, \textsl{\it The $f(q)$ mock theta function 
conjecture and partition ranks}, Inv. Math. 165, (2006), 243--266.


\bibitem{BO} 
K. Bringmann and K. Ono, \textsl{\it Arithmetic properties of 
coefficients 
of half-integral weight Maass-Poincar\'e series}, Math. Ann. 337 (2007), 
no. 3, 591--612.

\bibitem{C} 
J. L. Cardy, \textsl{\it Critical percolation in finite geometries}, J. 
Phys. A
{\bf 25}  201-206 (1992).

 \bibitem {C2}
J. Cardy, \textsl{\it Conformal invariance and percolation}, preprint (2001).   


 \bibitem {C3}
J. Cardy, \textsl{\it Linking numbers for self-avoiding loops and percolation: Application to
the spin quantum Hall transition}, Phys. Rev. Lett. {\bf 84} 3507-3510 (2000).  
 
\bibitem{CD} Y. Choie and N. Diamantis, 
\textsl{Rankin-Cohen brackets on higher-order modular forms},
Proceedings of the Bretton Woods workshop on Multiple
Dirichlet series (AMS Proceedings of Symposia in Pure Mathematics).

\bibitem{CDO} G. Chinta, N. Diamantis and C. O'Sullivan, \textsl{Second-order 
modular forms}, Acta Arith. 103 (2002), no. 3, 209--223.

\bibitem{D} 
J. Dub\'{e}dat, \textsl{\it Excursion decompositions for SLE and Watts' 
crossing formula}, Probab. Theory Relat. Fields {\bf 134},  453-488 (2006) (DOI: 
10.1007/s00440-005-0446-3).   

\bibitem{De} A. Deitmar, \textsl{Higher order group cohomology and the 
Eichler-Shimura map}, J. Reine Angew. Math. 629 (2009), 221--235

\bibitem{DJ} W. Duke and  P. Jenkins,
\textsl{Integral traces of singular values of weak Maass forms},
Algebra Number Theory 2 (2008), no. 5, 573--593.

\bibitem{DKMO} N. Diamantis, M. Knopp, G. Mason and C. O'Sullivan,  \textsl{
$L$-functions of second-order cusp forms}, Ramanujan J (2006) 12:327-347. 

\bibitem{DSi} N. Diamantis and D. Sim, \textsl{The classification of 
higher-order cusp forms}, J. Reine Angew. Math. 622 (2008), 121--153.

\bibitem{DSr} N. Diamantis and R. Sreekantan, \textsl{Iterated integrals 
and higher order automorphic forms}, Comment. Math. Helv. 81 (2006), no. 2, 
481--494.

\bibitem{GO} D. Goldfeld and C. O'Sullivan, \textsl{Estimating additive 
character sums for Fuchsian groups}, The Ramanujan Journal, 7 (2003), 
241--267.

\bibitem{HK} S. Husseini and M. Knopp, \textsl{Eichler cohomology and 
automorphic forms} Illinois J. Math. 15 1971 565--577.

\bibitem{HS}
C. Hongler and S. Smirnov, \textsl{\it in preparation}.

\bibitem{IM} \"O. Imamoglu and Y. Martin, \textsl{A converse theorem 
for second-order modular forms of level $N$},
Acta Arith. 123 (2006), no. 4, 361--376.

\bibitem{KZ} P. Kleban and D. Zagier, \textsl{Crossing Probabilities and 
Modular Forms}, Journal of Statistical Physics, 113 (2003), 431--454.

\bibitem{P}
H. Pinson, \textsl{\it Critical percolation on the torus}, J. Stat. Phys. 
{\bf 75} 1167-1177 (1994).

\bibitem{R} R. A. Rankin, \textsl{Modular forms and functions},
Cambridge University Press, Cambridge-New York-Melbourne, 1977.

\bibitem{S} 
S. Smirnov, \textsl{\it Critical percolation in the plane},  C. R. Acad. 
Sci. Paris S\'{e}r. I Math. {\bf 333} no. 3, 239-244 (2001).

\bibitem{SKZ} J. J. H. Simmons, P. Kleban and R. M. Ziff, 
\textsl{Percolation crossing formulas and conformal field theory},
J. Phys. A: Math. Theoret. 40, F771-F784 (2007).

\bibitem{W} 
G. Watts, \textsl{\it A crossing probability for critical percolation in 
two dimensions}, J. Phys. A: Math. Gen. {\bf  29},  L363-L368 (1996).  

\bibitem{Z} D. Zagier, \textsl{\it 
Traces of singular moduli} 
Motives, polylogarithms and Hodge theory, Part 
I (Irvine, CA, 1998), 211--244, Int. Press Lect.
Ser., 3, I, Int. Press, Somerville, MA, 2002.

\bibitem{Zi} R. M. Ziff, \textsl{\it On Cardy's formula for the critical 
crossing probability in 2D percolation} J. Phys. A 28, 1249-1255 
(1995); \textsl{\it Proof of crossing formula for 2D percolation} 
6479-6480 (1995).

\end{thebibliography}
\end{document}